\documentclass[10pt, two column, twoside]{IEEEtran}
\usepackage{amsthm}
\usepackage{amsmath}
\usepackage{algorithmic}
\usepackage[ruled]{algorithm2e}
\usepackage{mathrsfs}
\usepackage{graphicx}
\usepackage{subfigure}
\usepackage{cite}
\usepackage{bm,epsfig,amsthm,url}
\usepackage{indentfirst}
\usepackage{amssymb}
\usepackage{epstopdf}
\usepackage{xcolor}
\usepackage{mathtools}

\theoremstyle{definition}

\newtheorem{lemma}{Lemma}

\newtheorem{proposition}{Proposition}

\begin{document}
	
	\title{Active and Passive IRS Jointly Aided Communication: Deployment Design and Achievable Rate}
	\author{ Min~Fu,~\IEEEmembership{Member,~IEEE}~and~Rui~Zhang,~\IEEEmembership{Fellow,~IEEE}%
		\thanks{ 
			This work is supported in part by MOE Singapore under Award T2EP50120-0024, National University of Singapore under Research Grant R-261-518-005-720, and The Guangdong Provincial Key Laboratory of Big Data Computing.
			(\it{Corresponding author: Rui Zhang.})}
	\thanks{M. Fu is  with the Department of Electrical and Computer Engineering, National University of Singapore, Singapore 117583 (e-mail:  fumin@nus.edu.sg).
    R. Zhang is with the Chinese University of Hong Kong, Shenzhen, and Shenzhen Research Institute of Big Data, Shenzhen, China 518172 (e-mail: rzhang@cuhk.edu.cn). He is also with the Department of Electrical and Computer Engineering, National University of Singapore, Singapore 117583 (e-mail: elezhang@nus.edu.sg).}
}
	
\maketitle
	
\setlength\abovedisplayskip{2pt}
\setlength\belowdisplayskip{2pt}
\setlength\abovedisplayshortskip{2pt}
\setlength\belowdisplayshortskip{2pt}

\vspace{-6mm}
\begin{abstract}
In this letter, we study the wireless point-to-point communication from a transmitter (Tx) to a receiver (Rx), which is jointly aided by an active intelligent reflecting surface (AIRS) and a passive IRS (PIRS).
We consider two practical transmission schemes by deploying the two IRSs in different orders, namely, Tx$\rightarrow$PIRS$\rightarrow$AIRS$\rightarrow$Rx (TPAR) and Tx$\rightarrow$AIRS$\rightarrow$PIRS$\rightarrow$Rx (TAPR).
Assuming line-of-sight channels, we derive the achievable rates for the two schemes by optimizing the placement of the AIRS with the location of the PIRS fixed.
Our analysis shows that when the number of PIRS elements and/or the AIRS amplification power is small, the AIRS should be deployed closer to the Rx in both schemes, and TAPR outperforms TPAR with their respective optimized AIRS/PIRS placement.
Simulation results validate our analysis and show the considerable performance gain achieved by the jointly optimized 
AIRS/PIRS deployment over the existing benchmarks under the same power and IRS element budgets.
\end{abstract}

\begin{IEEEkeywords}
Intelligent reflecting surfaces (IRS), active IRS, double IRSs, IRS deployment, rate maximization. 
\end{IEEEkeywords}

\vspace{-3mm}
\section{Introduction}
 Intelligent reflecting surface (IRS) \cite{Wu2021Tutorial} has recently received significant attention from both academia and industry to improve the spectral and energy efficiency of future wireless networks cost-effectively.
 An IRS typically consists of an array of passive reflecting elements, each of which reflects the incident signal with a desired phase shift and/or amplitude.
 Unlike conventional active relays, passive IRSs (PIRSs) do not require costly transmit/receive radio-frequency (RF) chains and thus incur significantly lower power consumption \cite{di2020reconfigurable}. 
 As IRSs can be flexibly deployed in wireless networks to reconfigure wireless channels dynamically, they have been studied for achieving various functions such as coverage extension, rate enhancement,  interference mitigation, etc (see, e.g., \cite{Wu2021Tutorial} and the references therein).

 Most of the existing works (e.g., \cite{Wu2019intelligentJ, Ma2022Coverage, Han2022DoubleIRS,Zheng2021DoubleIRS,Fu2021RIS}) on IRS have considered the PIRS. 
 Equipped with passive loads (positive resistance), PIRS reflects the incident signal with desired phase shift and reflection gain no larger than one.
 In addition, PIRSs operate in the full-duplex mode without amplification/processing noise or self-interference \cite{Wu2021Tutorial}. 
 In particular, \cite{Wu2019intelligentJ} showed that deploying a single PIRS with $N$ reflecting elements can lead to a squared power scaling order (i.e., $\mathcal{O}(N^2)$) of the reflected signal, which is even higher than that  ($\mathcal{O}(N)$) of the active arrays.
 Furthermore, \cite{Han2022DoubleIRS} proposed a double-PIRS system where $N$ reflecting elements are equally allocated over two distributed PIRSs, which yields an even higher power scaling order ($\mathcal{O}(N^4)$) as compared to the single-PIRS system.
 Recently, wireless systems aided by multiple PIRSs with multiple signal reflections among them have also been investigated (see \cite{Mei2022Multireflection} and the references therein).
 Additionally, transmitting passive surfaces were proposed in \cite{Liu2022Compact}, where they are enclosed on the user's side to allow the incident signal penetration and provide passive beamforming gain for enhancing the performance.
 Despite the low cost of PIRSs, their signal coverage performance is severely limited by the product-distance path-loss of the cascaded multi-reflection channel\cite{Mei2022Multireflection}.  
 As a result, one PIRS needs to be allocated massive elements and/or deployed near the transmitter (Tx), the receiver (Rx), or other PIRSs in practice to make its signal reflection effective \cite{Mei2022Multireflection}.   
 
 \begin{figure}[t]
 	\centering
 	\includegraphics[scale=0.17]{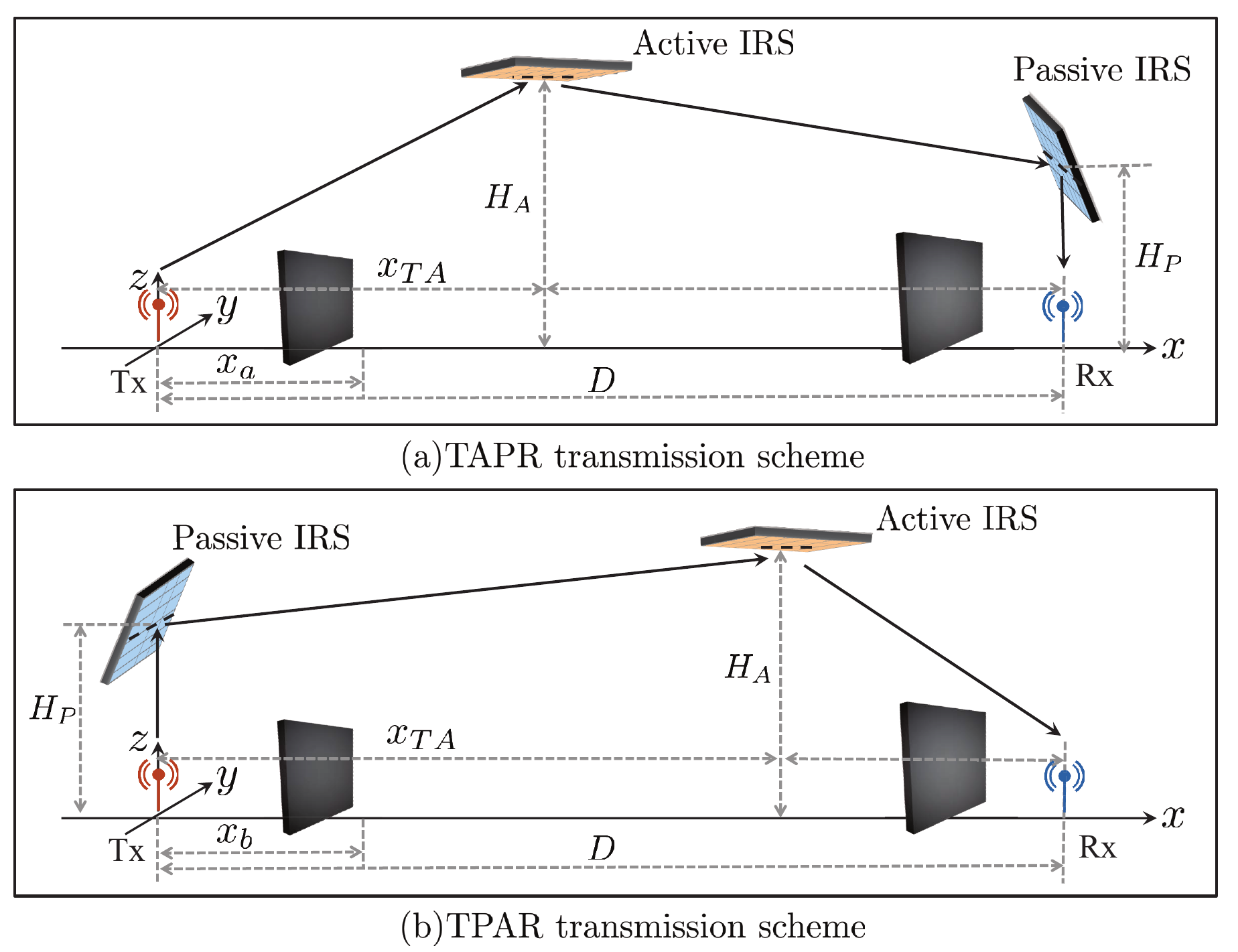}
 	\vspace{-3mm}		
 	\caption{The illustrations of AIRS-PIRS jointly aided wireless communication.} \label{Fig:systemmodel}
 	\vspace{-7mm}
 \end{figure}

 To address this issue with PIRS, the active IRS (AIRS) has emerged as a promising solution \cite{zhang2021active, Long2021active, You2021active, Khoshafa2021a, Liu2022Active, Zhi2022Active}.
 With negative resistance components connected to an additional power supply, AIRSs can adjust phase shifts as well as amplitude amplification with a gain exceeding one.
 Compared to PIRSs, although AIRSs offer an amplification gain (typically lower than that of active relays with dedicated RF amplifiers), they also induce non-negligible amplification noise in their reflected signals \cite{Long2021active}.
 In particular, given the IRS location and the total power budget, an AIRS was shown to achieve higher spectral efficiency\cite{Long2021active}, energy efficiency\cite{Liu2022Active}, and reliability\cite{Khoshafa2021a} than its PIRS counterpart.
 Furthermore, with optimized AIRS/PIRS placements, \cite{You2021active} showed that an AIRS only provides a power scaling of $\mathcal{O}(N)$ with $N$ active reflecting elements due to the amplification noise, but it achieves a higher rate than the PIRS when $N$ is small or the AIRS amplification power is high.
 Nevertheless, the AIRS system design (e.g., AIRS deployment and amplification factor optimization) is more complicated to mitigate noise amplification introduced by the AIRS.
 		Existing research on AIRS has mainly focused on comparing a single AIRS \cite{zhang2021active,Long2021active, You2021active, Khoshafa2021a, Liu2022Active, Zhi2022Active} or a hybrid IRS comprising active and passive elements \cite{Nguyen2021Hybrid} to a single PIRS.
 		However, the joint use of distributed AIRS and PIRS in practically challenging scenarios \cite{Mei2022Multireflection} such as multi-turn corridors in an indoor environment has not been studied yet.
 		
 To address the above problem, this letter proposes a new wireless communication system jointly aided by a pair of AIRS and PIRS to provide a double-reflection link between a single-antenna Tx and a single-antenna Rx via the two IRSs, while the other links are assumed to be blocked by obstacles.
 Additionally, we assume that the two IRSs can be properly deployed so that a cascaded line-of-sight (LOS) connection between the Tx and Rx can be established successively through them.
 Considering two deployment orders for the two IRSs, we present two practical transmission schemes, which are Tx$\rightarrow$PIRS$\rightarrow$AIRS$\rightarrow$Rx (TPAR) and Tx$\rightarrow$AIRS$\rightarrow$PIRS$\rightarrow$Rx (TAPR) schemes, respectively (see Fig. \ref{Fig:systemmodel}).
 Under the above settings, we aim to characterize the achievable rates for the two schemes with their respective optimized AIRS/PIRS placement.
 Our analytic results show that for both schemes, when the number of PIRS elements is small, the AIRS should be placed closer to the Rx with decreasing amplification power.
 Besides, the TAPR scheme generally outperforms the TPAR scheme when the number of PIRS elements and/or the AIRS amplification power is small.
 Simulation results validate our analysis and demonstrate that the jointly optimized AIRS/PIRS deployment achieves a much higher rate than the existing benchmarks under the same power and IRS element budgets.

\vspace{-3mm}
 \section{System  Model and Problem Formulation}\label{Section:model}
 As illustrated in Fig. \ref{Fig:systemmodel}, a single-antenna\footnote{
 			For the multiple-antenna case, the maximal ratio transmission/combining techniques can be applied at the Tx/Rx.} Tx transmits data to a single-antenna Rx\footnote{
 			This work can be extended to multi-user systems, but the corresponding multi-access design requires further investigation; an initial study on AIRS-aided multiple access is given in\cite{Chen2022active}.} that is $D$ meters (m) away.
Under a three-dimensional Cartesian coordinate system, the locations of the Tx and Rx are denoted by $\bm u_{T} = (0, 0, 0)$ and $\bm u_{R} = (D, 0, 0)$, respectively.
 A pair of PIRS and AIRS are deployed between them to assist in their communication.
 The PIRS and AIRS are respectively placed at given altitudes, denoted by $H_{P}$ and $H_{A}$, so that the Tx can communicate with the Rx via only a double-reflection link through the two IRSs over three LOS channels.
In particular, we assume that the PIRS is fixed above the Tx or Rx to minimize the path loss, and the AIRS can be flexibly placed between them to maximize the channel gain.
Thus, two transmission schemes\footnote{These two schemes can be practically used for downlink and uplink communications; however, the different deployment order of AIRS and PIRS generally results in different weighted sum-rate performance.} are considered by deploying the PIRS and AIRS in different orders, i.e., TAPR shown in Fig. 1(a) and TPAR shown in Fig. 1(b).

  Let $N_{p}$ and $N_{a}$ denote the number of passive/active reflecting elements on the PIRS/AIRS, respectively.
 We further define $\mathcal{N}_{a}$ and $\mathcal{N}_{p}$ as the sets containing all the elements on the AIRS and PIRS, respectively.
 Let $\bm \Theta_{P}=\text{diag}(\eta_{P,1} e^{j\phi_{P,1}},\ldots, \eta_{{P,N_{p}}} e^{j\phi_{P,N_{p}}})$ and $\bm \Phi_{A}=\text{diag}(\eta_{A,1} e^{j\phi_{A,1}},\ldots, \eta_{{A,N_{a}}} e^{j\phi_{A,N_{a}}})$ denote the reflection matrices of PIRS and AIRS respectively, where $\eta_{P,n}/\eta_{A,n}$ and $\phi_{P,n}/\phi_{A,n} \in [0, 2\pi)$\footnote{
 			The results in this letter can be extended to the practical setup with discrete phase-shift levels (e.g., by applying the quantization technique\cite{Fu2021RIS}).
 } denote respectively the passive/active reflection amplitude and passive/active phase-shift at element $n$.
 Therein, for maximal reflection of the PIRS and ease of practical implementation \cite{Wu2021Tutorial}, we further set $\eta_{P, n}=1, \forall n\in\mathcal{N}_{p}$.
 Furthermore, since the transmitted signal arriving at all active elements experiences the same path loss under LOS channels, a common amplification factor $\eta$ should be used to maximize the rate performance as in\cite{Long2021active, Liu2022Active}, i.e., $\eta_{A, n} = \eta\geq 1, \forall n\in \mathcal{N}_{a}$, which also simplifies the design of amplification factors of different IRS elements in practice.
 Moreover, different from the PIRS that reflects signals without incurring amplification noise, the AIRS generates non-negligible amplification noise at all reflecting elements, which is denoted by $\boldsymbol{n}_{\rm a}\in\mathbb{C}^{N_{a}\times 1}$ and assumed to follow the independent circularly symmetric complex Gaussian distribution, i.e., $\boldsymbol{n}_{\rm a}\sim \mathcal{CN}(\boldsymbol{0}_{N_{a}}, \sigma^2_{\rm F}{\bf I}_{N_{a}})$ with power $\sigma^2_{\rm F}$.

 \vspace{-3mm}
 \subsection{TAPR Transmission Scheme}
 First, we consider the TAPR scheme with the PIRS fixed above the Rx, as shown in Fig. 1(a). 
 \subsubsection{Signal model}
 Denoting the Tx-AIRS horizontal distance as $x_{TA}$(a design variable), the locations of the PIRS/AIRS are represented by $\bm u_{P} = (D, 0 , H_{P})$ and $\bm u_{A} = (x_{TA}, 0, H_{A})$, respectively.
 Thus, the distances between the Rx and PIRS, the Tx and AIRS, and the AIRS and PIRS are respectively given by
 \setlength\arraycolsep{2pt}
 \begin{eqnarray}
 	&&  d_{PR} =  H_{P},	d_{TA}(x_{TA}) = \sqrt{ x_{TA}^2 + H_{A}^2}, \\
 	&&	d_{AP}(x_{TA}) = \sqrt{(D - x_{TA})^2 + (H_{P}-H_{A})^2}.
 \end{eqnarray}

 		Let $\bm g_{\rm TA} \in \mathbb{C}^{N_{a}\times1}$, $\bm S_{\rm AP} \in \mathbb{C}^{N_{p}\times N_{a}}$, and $\bm h_{\rm RP}^{\sf H} \in \mathbb{C}^{1\times N_{p}}$ denote the baseband equivalent LOS channels for Tx$\rightarrow$AIRS, AIRS$\rightarrow$PIRS, and PIRS$\rightarrow$Rx links, respectively.
 Before modeling LOS channels, we first define the one-dimensional steering vector function for a uniform linear array (ULA) as 
 \setlength\arraycolsep{2pt}
 \begin{eqnarray}\label{Eq:steering-vector-ULA}
 	\bm {e}(\theta, M') = [1, e^{-j\pi\theta}, \ldots, e^{-j\pi(M' -1) \theta}]^{\sf T} \in \mathbb{C}^{M' \times 1},
 \end{eqnarray}
 where $\theta$ denotes the constant phase-shift difference between the signals at two adjacent elements and $M'$ denotes the size of the ULA. 
 For the uniform planar array (UPA) model, each array response vector can be expressed as the Kronecker product of two steering vector functions in the x-axis and y-axis directions, respectively.
 For example, the array response vector for the UPA from the Tx to AIRS is expressed as
 \begin{eqnarray}\label{Eq:steering-vector-UPA}
 	&&\!\!\!\!\!\bm a_{TA}(\theta_{TA},\vartheta_{TA}, N_{a}) = \bm {e}(\frac{2\Delta_A}{\lambda}\sin(\theta_{TA})\cos(\vartheta_{TA}), N_{a,x}) \otimes \nonumber\\
 	&& \hspace{4em}\bm {e}(\frac{2\Delta_A}{\lambda}\sin(\theta_{TA})\sin(\vartheta_{TA}), N_{a,y})
 	\in \mathbb{C}^{N_{a} \times 1},
 \end{eqnarray}
 where $\lambda$ denotes the signal wavelength, $\Delta_A$ denotes the element space at the AIRS, $(\theta_{TA},\vartheta_{TA})$ is the angle-of-arrival pair, and $N_{a,x}$ and $N_{a,y}$ denote the number of horizontal/vertical elements at the AIRS with $N_{a} = N_{a,x} \times N_{a,y}$.
 As such, $\bm S_{\rm PA}$, $\bm g_{\rm AT}$,  and $\bm h_{\rm PR}^{\sf H}$ are respectively given by
 \begin{eqnarray}\label{Eq:LoS model}
 \!\!\!\!\!\!\!\!\!\!	&&	\bm S_{\rm AP} \!=\! {\sqrt{\beta}}/{d_{AP} }e^{\frac{-j2\pi d_{AP} }{\lambda}} \bar{\bm s}_{AP} \tilde{\bm s}_{PA}^{\sf H}, \\
\!\!\!\!\!\!\!\!\!\!	&&	\bm g_{\rm TA} \!=\! {\sqrt{\beta}}/{d_{TA} }e^{\frac{-j2\pi d_{TA} }{\lambda}}\bar{\bm g}_{TA}, \bm h_{\rm RP}^{\sf H} \!=\! {\sqrt{\beta}}/{H_P}e^{\frac{-j2\pi H_p }{\lambda}}\bar{\bm h}_{RP}^{\sf H}, 
 \end{eqnarray}
 where $\beta$ denotes the reference channel power at a distance of 1 m, $\bar{\bm s}_{AP} = \bm a_{AP}(\theta_{AP},\vartheta_{AP}, N_{p})$,
 $\tilde{\bm s}_{PA} = \bm a_{PA}(\theta_{PA},\vartheta_{PA}, N_{a})$,
 $ \bar{\bm g}_{TA} = \bm a_{TA}(\theta_{TA},\vartheta_{TA}, N_{a}) $, and $ \bar{\bm h}_{RP}= \bm a_{RP}(\theta_{RP},\vartheta_{RP}, N_{p}) $.
In this letter, we assume perfect channel state information (CSI) is available\footnote{ 
 			The existing channel estimation techniques proposed for double-IRS-aided systems (see e.g., \cite{Mei2022Multireflection}) can be applied to obtain the CSI in our considered system.
 }.
 
 Though the Tx$\rightarrow$AIRS$\rightarrow$PIRS$\rightarrow$Rx reflecting link, the signal received at the Rx is given by 
 \begin{eqnarray}\label{Eq:receive-TAPR}
 	\!\!\!\!\!\! y_{a} &=& \bm h_{\rm RP}^{\sf{H}}\bm \Theta_{P}\bm S_{\rm AP}\eta\bm \Theta_{A}\Big(\bm g_{\rm TA}s + \bm n_{\rm a}\Big) +e \nonumber\\
 	\!\!\!\!\!\! &=& \underbrace{\bm h_{\rm RP}^{\sf{H}}\bm \Theta_{P}\bm S_{\rm AP}\eta\bm \Theta_{A}\bm g_{\rm TA}s}_{\text{Double-reflected signal}} + \underbrace{\bm h_{\rm RP}^{\sf{H}}\bm \Theta_{P}\bm S_{\rm AP}\eta\bm \Theta_{A} \bm n_{\rm a}}_{\text{Noise introduced by the AIRS}} + e,
 \end{eqnarray}
 where $s$ denotes the transmitted signal with power $P_{\rm t}$ of the Tx, and $e$ denotes the additive white Gaussian noise at the Rx with power $\sigma^2$.
 As observed from \eqref{Eq:receive-TAPR}, the two IRSs can jointly enhance the double-reflected signal power whereas the AIRS introduces additional amplification noise.
 The received signal-to-noise ratio (SNR) at the Rx is expressed as
 \begin{eqnarray}\label{SNR-TAPR}
 	\!\!{\sf SNR}_{a}(\bm \Theta_{A},\bm \Theta_{P}, \eta, x_{TA}) = \frac{\|\bm h_{\rm RP}^{\sf{H}}\bm \Theta_{P} \bm S_{\rm AP}\eta\bm \Theta_{A}\bm g_{\rm TA}\|^2P_{\rm t}}{\|\bm h_{\rm RP}^{\sf{H}}\bm \Theta_{P}\bm S_{\rm AP}\eta\bm \Theta_{A}\|^2\sigma^2_{\rm F} + \sigma^2},
 \end{eqnarray}
 and  the achievable rate (in bits per second per Hertz or bps/Hz)  for the TAPR scheme is given by
 \begin{eqnarray}\label{rate-TAPR}
 	{\sf R}_{a} = \log_2\left(1+ {\sf SNR}_{a}(\bm \Theta_{A},\bm \Theta_{P}, \eta, x_{TA})\right).
 \end{eqnarray}
		To ensure its reflection amplifier operating in an amplification mode \cite{Amato2018Tunneling} so as to improve the performance over PIRS,
		the AIRS should be properly deployed to meet the following constraint:
		\begin{eqnarray}
			\eta \geq 1.\label{Cons:activeIRS1-TAPR}
		\end{eqnarray}
 Denoting $P_{\rm F}$ as the given amplification power budget of the AIRS, the amplification power constraint over the signal reflected by the AIRS is given by \cite{You2021active,Long2021active}
 \begin{eqnarray}
 		\eta^2\left(P_{\rm{t}} \| \bm \Theta_A \bm g_{\rm TA}\|^2+\sigma_{\rm F}^2 \Vert { \bm\Theta_A} {\bf I}_{N_{a}}\|^2\right) \le P_{\rm F}. \label{Cons:activeIRS2-TAPR}
 \end{eqnarray}
\subsubsection{Problem formulation}
We aim to maximize the achievable rate for the Rx by jointly optimizing the reflection matrices at the PIRS/AIRS (i.e., $\bm \Theta_{P}$ and $\eta\bm \Theta_{A}$) and the AIRS deployment (i.e., $x_{TA}$).
Accordingly, for the TAPR scheme, the optimization problem is formulated as
\begin{subequations}\label{problem:TAPR}
	\begin{eqnarray}
		\hspace{-5mm} \mathop { \text{maximize} }\limits_{\substack{\bm \Theta_{A},\bm \Theta_{P}, \eta,\\ x_{TA} }}&& \log_2\left(1+ {\sf SNR_a}(\bm \Theta_{A},\bm \Theta_{P}, \eta, x_{TA})\right) \nonumber \\
		\hspace{-5mm} \text{subject to}
		&& |\bm \Theta_{P}(n,n)| = 1, \forall n \in {{\mathcal{N}_p}}, \label{Cons:PIRS-phase} \\
		&& |\bm \Theta_{A}(n,n)| = 1, \forall n \in {{\mathcal{N}_a}}, \label{Cons:AIRS-phase}\\
		&& 0	\leq x_{TA}\leq D,  \\
		&& \text{Constraints~}
		 \eqref{Cons:activeIRS1-TAPR},~ \eqref{Cons:activeIRS2-TAPR}. \nonumber 
	\end{eqnarray}
\end{subequations}
Since the objective function is increasing with $\eta$, the amplification power constraint in \eqref{Cons:activeIRS2-TAPR} needs to be active at the optimal solution of problem \eqref{problem:TAPR}.
Furthermore, constraint \eqref{Cons:activeIRS2-TAPR} is independent of the phase values in $\bm \Theta_{A}$ and $\bm \Theta_{P}$.
		In addition, since the objective function is increasing with $\|\bm h_{\rm RP}^{\sf{H}}\bm \Theta_{P}  \bar{\bm s}_{\rm AP}\|^2$ and its denominator is independent of the phase value of $\bm \Theta_{A}$, it is maximized by designing the phases of the AIRS and PIRS to align in the cascaded Tx-AIRS-PIRS-Rx channel.
Thus, for given $x_{TA}$, the optimal phase design of the AIRS/PIRS and amplification factor are respectively given by
\begin{eqnarray}
	&& \bm \Theta_{A}(n,n) = e^{-j(\angle[ \tilde{\bm s}_{PA}^{\sf H}]_n + \angle[\bm g_{\rm TA}]_n)}, \label{solution:active-phase2} \\
	&& \bm \Theta_{P}(n,n) = e^{-j(\angle[\bar{\bm s}_{AP}]_n + \angle[\bm h^{\sf H}_{\rm RP}]_n )}, \label{solution:passive-phase2}\\
	&& 
	\eta = \sqrt{P_{\rm F}d^2_{TA}(x_{TA})/(P_{\rm{t}}N_a\beta +d^2_{TA}(x_{TA})\sigma_{\rm F}^2 N_a}). \label{solution:amflifyfactor2}
\end{eqnarray}
\begingroup
\allowdisplaybreaks

Combining \eqref{solution:active-phase2}, \eqref{solution:passive-phase2}, \eqref{solution:amflifyfactor2} into \eqref{SNR-TAPR} yields
\begin{eqnarray} \label{Eq:TAPR-SNR-dis}
	&&\!\!\!\!\!{\sf SNR}_{a}(x_{TA}) \nonumber \\
	&& \!\!\!\!\!= \frac{P_{\rm{t}}P_{\rm F}\beta^2N_{a}N_{p}^2}{C_{1}d^2_{TA}(x_{TA})+ C_{2}d^2_{AP}(x_{TA})+C_{3} d^2_{TA}(x_{TA}) d^2_{AP}(x_{TA}) }, \nonumber\\
\end{eqnarray}
where $C_{1} = \sigma_{F}^2P_{\rm F}\beta N_{p}^2$, $C_{2} = \sigma^2P_{\rm{t}}H_P^2$, and $C_{3} = H_P^2\sigma_{\rm F}^2\sigma^2/\beta$.
After substituting \eqref{solution:active-phase2}, \eqref{solution:passive-phase2}, and \eqref{solution:amflifyfactor2} into \eqref{problem:TAPR}, problem \eqref{problem:TAPR} is equivalent to the following, 
\begin{subequations}\label{problem:TAPR-reduced}
	\begin{eqnarray}
		\mathop{ \text{maximize}}\limits_{x_{TA} } && 
		\log_2(1+{\sf SNR}_{a}(x_{TA}) )\nonumber \\
		\text{subject to}
		&&  x_{a} \leq x_{TA} \leq D,	\label{Cons:deployment1}	
	\end{eqnarray}
\end{subequations}
where constraint \eqref{Cons:deployment1} corresponds to  \eqref{Cons:activeIRS1-TAPR} in problem \eqref{problem:TAPR}, 
\begin{eqnarray}
	x_{a} = \sqrt{\max\{0, N_{a}\beta P_{\rm{t}}/( P_{\rm F}- N_{a}\sigma_{\rm F}^2)-H_{A}^2\}}.
\end{eqnarray} 
Although the optimal solution to problem \eqref{problem:TAPR-reduced} is difficult to be characterized in closed form due to the complicated expression of ${\sf SNR}_{a}(x_{TA})$, its value can be obtained by using a one-dimensional search over $x_{TA} \in [x_{a}, D]$.
In the following, we first characterize the effects of AIRS amplification power and the number of PIRS elements on the optimal AIRS placement.
\begin{lemma}\label{lemma:TAPR-distance}
	The optimal Tx-AIRS horizontal distance (i.e., $x^{\star}_{TA}$) to problem \eqref{problem:TAPR-reduced} is \textit{non-increasing} with $P_{\rm F}$ and monotonically \textit{decreasing} with $N_{p}$.
\end{lemma}
\begin{proof}
	First, it can be shown that $x_{a}$ in \eqref{Cons:deployment1} is non-increasing with $P_{\rm F}$.
	In addition, it is observed that $C_{1}$ and $d_{TA}(x_{TA})$ in ${\sf SNR}_{a}(x_{TA})$ monotonically increase with both $P_{\rm F}$ and $x_{TA}$ while $d_{AP}(x_{TA})$ in ${\sf SNR}_{a}(x_{TA})$ decreases with $x_{TA}$. 
	As a result, supposing without constraint \eqref{Cons:deployment1}, when $P_{\rm F}$ increases, the optimal solution $x^{\star}_{TA}$ should be decreased to guarantee the first derivative of the objective function to be zero.
	Combining the above, $x^{\star}_{TA}$ is non-increasing with $P_{\rm F}$.
	As for $N_{p}$, it can be shown that $x_{a}$ is independent of $N_{p}$ while $C_{1}$ monotonically increases with $N_{p}$.
    Similarly, it can be concluded that $x^{\star}_{TA}$ is decreasing with $N_{p}$. 
		This thus completes the proof.
\end{proof}

First, as observed from \eqref{solution:amflifyfactor2}, amplification factor $\eta$ is proportional to $x_{TA}$ and $P_{\rm F}$.
Based on \eqref{SNR-TAPR}, although the received signal power increases with $\eta$ and $N_{p}$, the AIRS-induced noise is also increased.
Thus, from Lemma \ref{lemma:TAPR-distance}, as $P_{\rm F}$ and $N_{p}$ increase, the AIRS has to be deployed closer to the Tx to attenuate the AIRS-PIRS channel gain, thus suppressing its induced noise power at the Rx.

\vspace{-3mm}
\subsection{TPAR Transmission Scheme}
Next, we consider TPAR with the PIRS  fixed above the Tx, as shown in Fig. \ref{Fig:systemmodel}(b). Similarly to TAPR, denoting the horizontal distance between the Tx and AIRS as $x_{TA}$, the locations of the PIRS/AIRS are represented by $\bm u_{P} = (0, 0, H_{P})$ and $\bm u_{A} = (x_{TA}, 0, H_{A})$, respectively.
The signal model of TPAR is also similar to that of TAPR, while the expressions of the AIRS amplification noise and its power constraint are different since the order of the AIRS/PIRS placements is reversed.
The detailed expressions are omitted for brevity.

Similarly to the design of $\bm \Theta_{A}$ and $\bm \Theta_{P}$ in TAPR, the amplification factor and receiver SNR in TPAR are given by
\begin{eqnarray}
	&&\eta = \sqrt{P_{\rm F}d^2_{AP}/(P_{\rm{t}}N_a\beta^2N_{p}^2/H_{P}^2 +d^2_{AP}\sigma_{\rm F}^2 N_a)}, \label{solution:eta-TPAR} \\
	&&\!\!\!\!\!{\sf SNR}_{b}(x_{TA}) \nonumber \\
	=&&\!\!\!\!\! \frac{N_a\beta^2 N_{p}^2 P_{\rm t}P_{\rm F}}{C_{2}/ad^2_{AP}(x_{TA}) \!+\!aC_{1} d^2_{AR}(x_{TA}) \!+\!C_{3}d^2_{AP}(x_{TA})d^2_{AR}(x_{TA}) }, \nonumber \label{Eq:TPAR-SNR-dis} \\
\end{eqnarray}
where $a = \frac{P_t \sigma^2}{P_{\rm F} \sigma_{\rm F}^2}$.
As a result, the optimization problem of the TPAR scheme is formulated as
\begin{subequations}\label{problem:TPAR-reduced}
	\begin{eqnarray}
		\mathop { \text{maximize} }\limits_{\ x_{TA} }&& \log_2(1+{\sf SNR}_{b}( x_{TA})) \nonumber \\
		\text{subject to}
		&& x_{b} \leq x_{TA} \leq D, \label{Cons:TPAR-deployment}	
	\end{eqnarray}
\end{subequations}
where constraint \eqref{Cons:TPAR-deployment} is due to the amplification power constraint at the AIRS, and 
\begin{eqnarray}
	\! x_{b} \!=\! \sqrt{\!\max\{0, N_{a}N_{p}^2\beta^2 P_{\rm{t}}/(( P_{\rm F}\!-\! N_{a}\sigma_{\rm F}^2)H_P^2) \!-\!(H_P\!-\!H_A)^2\}}. \nonumber
\end{eqnarray}
Similarly to problem \eqref{problem:TAPR-reduced}, the optimal solution of problem \eqref{problem:TPAR-reduced} can also be obtained by using the one-dimensional search over $x_{TA} \in [x_{b}, D]$.
Likewise, we have the following lemma for the TPAR scheme.
\begin{lemma}\label{lemma:TPAR-distance}
	The optimal Tx-AIRS horizontal distance (i.e., $x^{\star}_{TA}$) to problem \eqref{problem:TPAR-reduced} is \textit{non-increasing} with $P_{\rm F}$ and monotonically \textit{increasing} with $N_{p}$.
\end{lemma}
\begin{proof}
	The proof is similar to that of Lemma \ref{lemma:TAPR-distance}, and thus omitted for brevity.
\end{proof}
Lemma \ref{lemma:TPAR-distance} can be intuitively understood for the TPAR scheme.
The AIRS-induced noise is independent of the number of PIRS reflecting elements.
As observed from \eqref{solution:eta-TPAR}, with a smaller $P_{\rm F}$ and/or a larger $N_p$, the AIRS should be deployed closer to the Rx to provide a higher amplification factor with $\eta >1$ and reduce the path loss of the cascaded channel, thus increasing the signal power at the Rx.

\vspace{-4mm}
\section{Low-Complexity Placement Design and Performance Comparison} \label{Section:solution}
In this section, we present suboptimal AIRS deployment solutions in closed-form for the TAPR and TPAR schemes, respectively.
Based on these results, we compare the performance of the two schemes in terms of achievable rate.

\vspace{-4mm}
\subsection{Suboptimal Solutions}
First, to simplify problems \eqref{problem:TAPR-reduced} and  \eqref{problem:TPAR-reduced}, we present a useful lemma as follows (for which the proof is omitted for brevity).
\begin{lemma}\label{lemma:TPAR-TAPR}
	Given $D \gg \max\{H_{A}, H_{P}\}$, we have
	$C_{1}d^2_{TA}(x_{TA})+ C_{2}d^2_{AP}(x_{TA})\gg C_{3} d^2_{TA}(x_{TA}) d^2_{AP}(x_{TA})$ and $C_{2}/ad^2_{AP}(x_{TA})+aC_{1} d^2_{AR}(x_{TA}) \gg C_{3}d^2_{AP}(x_{TA})d^2_{AR}(x_{TA}) $, $\forall x_{TA}\in [0, D]$, if
	\begin{eqnarray} \label{Condition:distance}
\!\!\!\!	\min\{\!\frac{\sqrt{\beta} N_{p}\sqrt{P_{\rm t}\beta}}{H_p\sigma_{\rm F} } \!+\! \frac{\sqrt{P_{\rm F}\beta}}{\sigma}, \frac{\sqrt{\beta} N_{p}\sqrt{P_{\rm F}\beta}}{H_p\sigma } \!+\! \frac{\sqrt{P_{\rm t}\beta}}{\sigma_{\rm F}}\} \!\gg\! D. 
	\end{eqnarray}
\end{lemma}

\subsubsection{TAPR Scheme}
When condition \eqref{Condition:distance} is satisfied, $ \sf{R}_a$ in problem \eqref{problem:TAPR-reduced} can be approximated as $\overline{\sf{R}}_a = \log_2(1+\overline{ \sf{SNR}}_a)$, where
$\overline{ \sf{SNR}}_a = \frac{P_{\rm{t}}P_{\rm{F}}\beta^2N_{a}N_{p}^2}{C_{1}x^2_{TA} + C_{2}(D-x_{TA})^2 }$.
Therefore, problem \eqref{problem:TAPR-reduced} is approximated as
\begin{eqnarray}\label{problem:TAPR-relaxed}
	\mathop { \text{maximize} }\limits_{ x_{TA} } && \log_2\Big(1+\frac{P_{\rm{t}}P_{\rm{F}}\beta^2N_{a}N_{p}^2}{C_{1}x^2_{TA} + C_{2}(D-x_{TA})^2 }\Big) \nonumber \\
	\text{subject to}
	&& \text{constraint~}\eqref{Cons:deployment1}.
\end{eqnarray}
By setting the first derivative of the objective function of \eqref{problem:TAPR-relaxed} with respect to $x_{TA}$ to zero and then projecting it into the feasible set of problem \eqref{problem:TAPR-relaxed}, the optimal solution to problem \eqref{problem:TAPR-relaxed} can be obtained as
\begin{eqnarray}\label{Solution:TAPR-x}
	x^{a}_{TA} 
	  =  \max\{\frac{ aD}{ a + b}, x_{a} \},
\end{eqnarray}
where $a = \frac{P_t \sigma^2}{P_{\rm F} \sigma_{\rm F}^2}$ and $b = \frac{\beta N_{p}^2}{H_P^2}$.
Note that when $\frac{ aD}{ a + b} > x_{a}$, the amplification factor (i.e., $\eta$) in problem \eqref{problem:TAPR} is larger than 1. 
Given $\eta> 1$, the receiver SNR is approximated as
\begin{eqnarray}
\!\!\!\!\!\!\!\!\!\!\!\!\!	&& \overline{ \sf{SNR}}_{a} = \frac{P_{\rm F} \beta N_{a}}{D^2\sigma^2}(a+b) =\frac{\beta N_{a}}{D^2}(\frac{P_{\rm t}  }{\sigma_{\rm F}^2}\!+\!\frac{P_{\rm F}N_{p}^2 \beta }{\sigma^2H_{p}}). \label{SNR:TAPR-appro}
\end{eqnarray}

\subsubsection{TPAR Scheme}
Similarly, based on Lemma \ref{lemma:TPAR-TAPR}, 
problem \eqref{problem:TPAR-reduced} is approximated as
\begin{eqnarray}\label{problem:TPAR-relaxed}
	\mathop { \text{maximize} }\limits_{ x_{TA} }&& \log_2\Big(1+\frac{P_{\rm t}P_{\rm F}\beta^2N_{a}N_{p}^2}{1/aC_{2}x_{TA}^2+ aC_{1}(D-x_{TA})^2}\Big) \nonumber \\
	\text{subject to}
	&&\text{constraint~}\eqref{Cons:TPAR-deployment}.
\end{eqnarray}
Similarly to problem \eqref{problem:TAPR-relaxed}, the optimal solution to problem \eqref{problem:TPAR-relaxed}  is obtained as
\begin{eqnarray}\label{Solution:TPAR-x}
	x^{b}_{TA} 
	  = \max\{\frac{ab D}{ ab + 1}, x_{b} \}.
\end{eqnarray}
Given $\eta> 1$, the approximated receiver SNR is given by
\begin{eqnarray} 
\!\!\!\!\!\!\!\!\!\!\!\!\!	&& \overline{ \sf{SNR}}_{b} =  \frac{P_{\rm F} \beta N_{a}}{D^2\sigma^2}(ab+1) = \frac{\beta N_{a}}{D^2}(\frac{P_{\rm t} N_{p}^2 \beta }{\sigma_{\rm F}^2H_{p}}+\frac{P_{\rm F} }{\sigma^2}). \label{SNR:TPAR-appro}
\end{eqnarray}

\vspace{-4mm}
\subsection{Comparisons between TPAR and TAPR}
With the suboptimal AIRS deployments given in \eqref{Solution:TAPR-x} and \eqref{Solution:TPAR-x}, the achievable rates for the TPAR and TAPR schemes are compared as follows.

\begin{figure*}[htbp!]
	\centering
	\begin{minipage}{.28\textwidth}
		\centering		
		\includegraphics[scale=0.34]{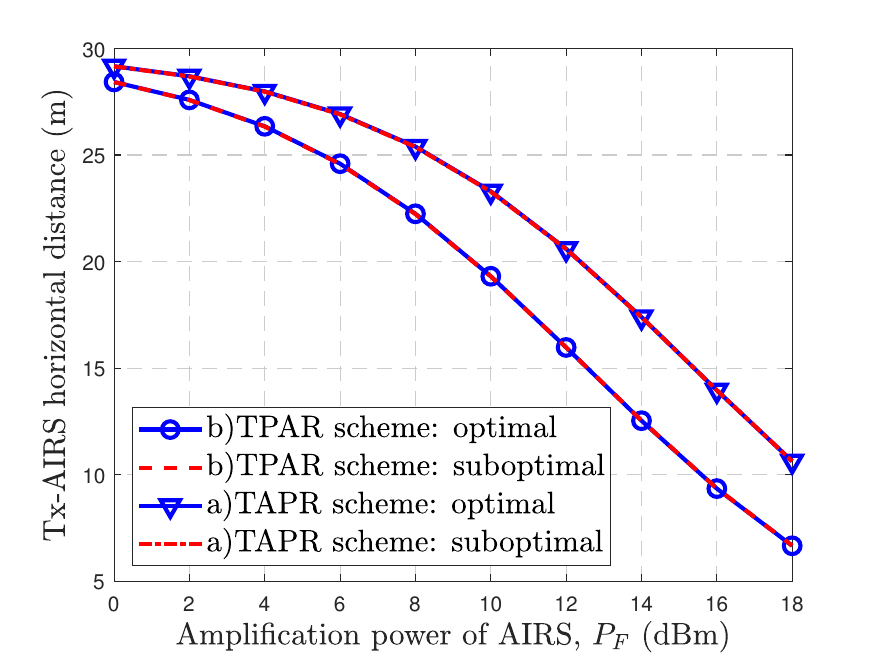}	
		\vspace{-3mm}		
		\caption{Optimal AIRS placement versus amplification power of the AIRS.} \label{Fig:element-distance}	
		\vspace{-6mm}
	\end{minipage}
	\hspace{2mm}
	\begin{minipage}{.28\textwidth}
		\centering
		\includegraphics[scale=0.34]{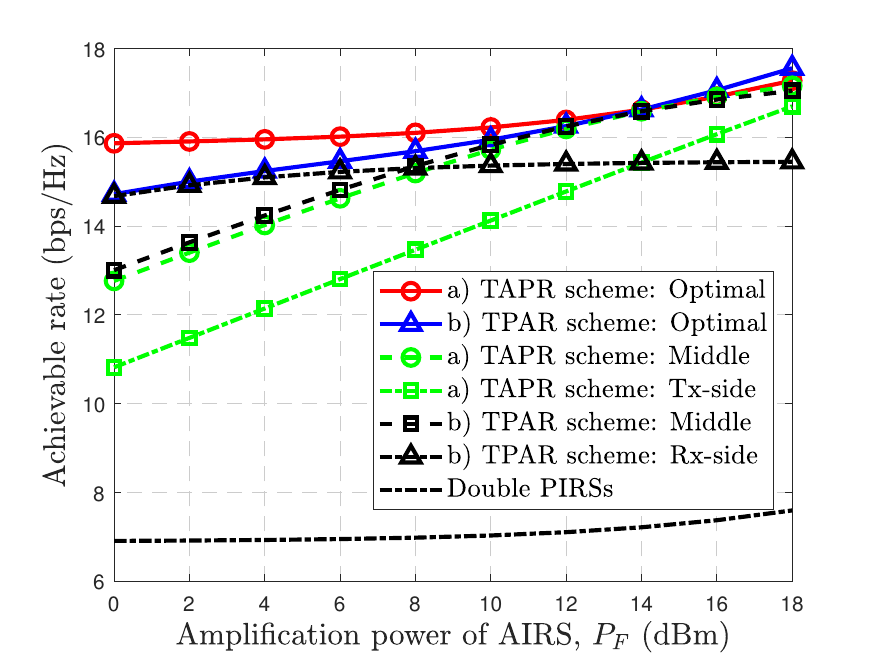}	
		\vspace{-3mm}	
		\caption{Achievable rate versus amplification power of the AIRS.} \label{Fig:P_F-rate}
		\vspace{-6mm}
	\end{minipage}
\hspace{2mm}
\begin{minipage}{.28\textwidth}
	\centering
	\includegraphics[scale = 0.34]{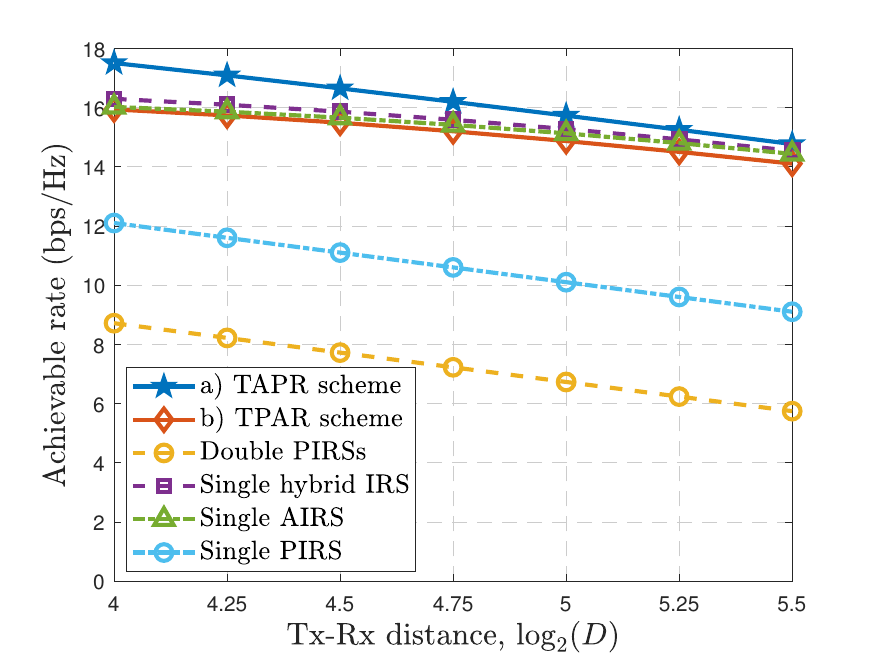}
	\vspace{-3mm}		
	\caption{Achievable rate versus Tx-Rx distance when $P_F = 2$ dBm.}\label{fig:distance}	
	\vspace{-6mm}
\end{minipage}
\end{figure*}

\begin{proposition}\label{proposition:TPARversusTATR}
Given $\eta > 1$ and $N_{p}\leq\frac{H_{P}}{\sqrt{\beta}}$, the achievable rate of the TPAR scheme is no less than that of the TAPR scheme, i.e., $\overline{ \sf{R}}_a \geq \overline{ \sf{R}}_b$, if 
\begin{eqnarray}
	P_{\rm F} \leq \frac{P_{\rm t}\sigma^2}{\sigma_{\rm F}^2}.
\end{eqnarray}
Otherwise, $\overline{ \sf{R}}_a < \overline{ \sf{R}}_b$.
The result is reversed when $N_{p}\geq\frac{H_{P}}{\sqrt{\beta}}$.
\end{proposition}
\begin{proof}
Comparing $\overline{ \sf{R}}_a$ and $\overline{ \sf{R}}_{b}$ is equivalent to comparing $\overline{ \sf{SNR}}_a$ and $\overline{ \sf{SNR}}_{b}$. 
Thus, with $\eta > 1$, based on \eqref{SNR:TAPR-appro} and \eqref{SNR:TPAR-appro}, we have
\begin{eqnarray} \label{proof:Eq-SNR}
\overline{ \sf{SNR}}_{a} - \overline{ \sf{SNR}}_b 
	&=&	\frac{P_{\rm F} \beta N_{a}}{D^2\sigma^2}(a+b-ab-1) \nonumber \\
	&=& \frac{ \beta N_{a} }{D^2}\frac{(\frac{P_{\rm t}}{\sigma_{\rm F}^2}-\frac{P_{\rm F}}{\sigma^2})(H_p^2-N_{p}^2\beta)
	}{H_p^2}.  
\end{eqnarray}	
Based on \eqref{proof:Eq-SNR}, we can conclude that $\overline{ \sf{R}}_a\geq \overline{ \sf{R}}_{b}$, if 
\begin{eqnarray}
	N_{p}\leq\frac{H_{P}}{\sqrt{\beta}}, \frac{P_{\rm F}}{\sigma^2} \leq \frac{P_{\rm t}}{\sigma_{\rm F}^2}, \text{or~} N_{p}\geq\frac{H_{P}}{\sqrt{\beta}}, \frac{P_{\rm F}}{\sigma^2} \geq \frac{P_{\rm t}}{\sigma_{\rm F}^2}.
\end{eqnarray}	
Otherwise, $\overline{ \sf{R}}_a < \overline{ \sf{R}}_b$.
This thus completes the proof.	
\end{proof}

\vspace{-4mm}
\section{Numerical Results}\label{Section:results}
In this section, we present numerical results to compare the rate performance of the TAPR and TPAR schemes in an AIRS-PIRS jointly aided communication system.
If not specified otherwise, the simulation parameters are set as follows. 
The AIRS and PIRS are assumed to be deployed at altitudes of $H_{A} = 6$ m and $H_{P} = 5$ m, respectively.
The carrier frequency is 3.5 GHz and thus the reference channel gain is $\beta = (\lambda/4\pi)^2 = -43$ dB, with the carrier wavelength $\lambda = 0.087$ m.	
Other parameters are set as $D = 30$ m, $P_{\rm t} = 20$ dBm, $\sigma^2 = -80$ dBm, and $\sigma_{\rm F}^2 = 4\sigma^2 = -74$ dBm.
Under our simulation setup, we consider a limited budget of reflecting elements with  $N_{a} < N_{p} < \frac{H_{P}}{\sqrt{\beta}}\approx 700$ (see Proposition 1). 
The numbers of AIRS and PIRS elements are thus fixed as $N_{a} = 450$ and $N_{p} = 600$, respectively.

For performance comparison, we consider the existing benchmarks with a single AIRS \cite{You2021active}, a hybrid IRS (with active and passive elements) \cite{Nguyen2021Hybrid}, double PIRSs \cite{Han2022DoubleIRS}, and a single PIRS \cite{Fu2021RIS}.
For the double-PIRS scheme, both PIRSs are equipped with $(N_{a} + N_{p})/2$ passive elements and deployed above the Tx and Rx, respectively.
As for the single-PIRS scheme, the PIRS is equipped with $N_{a} + N_{p}$ passive elements and deployed above the Tx or Rx.
Furthermore, the transmit power of Tx in the double-PIRS/single-PIRS systems is $(P_{\rm t} + P_{\rm F})$.
As such, it is fair to compare it with the AIRS-PIRS jointly aided system.
For the latter, we also consider four fixed AIRS/PIRS deployments, where $x_{AP}$ equals $D/2$ or $D$ in both TAPR/TPAR schemes.

Fig. \ref{Fig:element-distance} shows the optimal Tx-AIRS horizontal distance versus AIRS amplification power.
We observe that the optimal AIRS placement in both schemes is closer to the Rx (i.e., larger $x_{TA}$) as $P_{\rm F}$ decreases. 
This is consistent with Lemmas \ref{lemma:TAPR-distance} and \ref{lemma:TPAR-distance}.
Additionally, the proposed suboptimal solution achieves near-optimal performance, validating its effectiveness and the result in Lemma 3.

Fig. \ref{Fig:P_F-rate} compares the achievable rates versus AIRS amplification power.
First, as $P_{F}$ increases, we observe that the rate performance of the AIRS-PIRS jointly aided system increases due to the increasing amplification gain from AIRS.
However, compared to these fixed deployment strategies, the AIRS deployment design dynamically balances the signal and noise amplification at the AIRS, thereby further improving performance for both TAPR/TPAR schemes.
Furthermore, with the optimal AIRS placement, TAPR performs better than TPAR when $P_{\rm F} < \frac{P_{\rm t}\sigma^2}{\sigma_{\rm F}^2}$, because it suffers less path loss of the cascaded channel and thus results in higher received power.
These results are consistent with Proposition \ref{proposition:TPARversusTATR}.
In addition, thanks to the optimal AIRS/PIRS placement and the resulting amplification gain at AIRS, the AIRS-PIRS jointly aided system achieves a much higher rate than the baseline double-PIRS system under a limited IRS element budget.

		Fig. \ref{fig:distance} shows that the achievable rate of different schemes with optimal IRS placement versus Tx-Rx distance for $P_F = 2$ dBm.		
		We can observe that TAPR performs better than hybrid-IRS and single-AIRS benchmarks, especially when the Tx-Rx distance is relatively short.		
		This is because TAPR provides higher multiplicative beamforming and amplification gains by jointly using PIRS and AIRS as well as an additional path to attenuate amplification noise as compared to the two single-reflection schemes.		
		This indicates that deploying distributed AIRS and PIRS with the best ordering can be advantageous over combining them into one hybrid IRS or using AIRS only.

\section{Conclusion}\label{Section:conclusion}
In this letter, we studied a new AIRS-PIRS jointly aided wireless communication system. 
Under the LOS channel model, we analyzed the effects of the AIRS deployment in two transmission schemes (i.e., TAPR and TPAR) on their achievable rates.
Our analysis revealed that the AIRS should be deployed closer to the Rx in both schemes, and TAPR achieves a higher rate than TPAR when the number of PIRS elements and/or AIRS amplification power is limited. 
Simulation results validated our analysis and demonstrated that our proposed system considerably outperforms the existing benchmarks in terms of achievable rate, and the performance gain highly depends on the AIRS/PIRS deployment.

\bibliographystyle{IEEEtran}
\bibliography{ref} 

\end{document}